\documentclass[a4paper,USenglish,cleveref, autoref, thm-restate]{lipics-v2021}
\usepackage{amsmath,amssymb,amsthm,bbold,mathtools}
\usepackage{url}
\usepackage{hyperref}
\usepackage{doi}

\usepackage{graphicx}
\usepackage{caption}
\usepackage{subcaption}
\usepackage{float}

\usepackage{lineno}
\nolinenumbers 

\numberwithin{equation}{section}

\renewcommand{\H}{\mathbb{H}^2}
\newcommand{\SSS}{\mathbb{S}}

\newcommand{\proj}{\rho} 

\newcommand{\h}[1]{\widetilde{#1}} 

\newcommand{\hypS}{S} 

\newcommand{\id}{\mathbb{1}} 

\newcommand{\gp}{\Gamma}
\newcommand{\g}{\gamma}
\newcommand{\alphabet}{\mathcal{A}_\Gamma}
\newcommand{\inv}[1]{{#1}^{\scriptscriptstyle -1}} 

\newcommand{\poly}{\Pi} 
\newcommand{\prev}[1]{\mathsf{prev}(#1)}
\newcommand{\nxt}[1]{\mathsf{next}(#1)} 
\newcommand{\source}[1]{\mathsf{source}(#1)}
\newcommand{\target}[1]{\mathsf{target}(#1)} 
\newcommand{\pair}[1]{\mathsf{pair}(#1)} 
\newcommand{\point}[1]{\mathsf{point}(#1)}
\newcommand{\alp}[1]{\mathsf{w}(#1)}

\newcommand{\graph}{G_\Pi} 
\newcommand{\tree}{\mathcal{T}} 

\newcommand{\pat}[1]{\mathtt{p}(#1)}

\newcommand{\Dir}[1]{\mathcal{D}_{\h{#1}}} 
\newcommand{\fd}{\mathcal{F}} 

\DeclareMathOperator{\dia}{diam}

\DeclareMathOperator{\inter}{int}

\newcommand{\tab}{\texttt{t}}

\title{Computing a Dirichlet domain for a hyperbolic surface}

\author{Vincent Despr\'e}{Universit\' e de Lorraine, CNRS, Inria, LORIA, F-54000 Nancy, France}{vincent.despre@loria.fr}{}{}

\author{Benedikt Kolbe}{Hausdorff Center for Mathematics, University of Bonn, Germany \and \url{https://hyperbolictilings.wordpress.com/}}{benedikt.kolbe@physik.hu-berlin.de}{}{This work was done while this author was working at Universit\' e de Lorraine, CNRS, Inria, LORIA, F-54000 Nancy, France}

\author{Hugo Parlier}{Department of Mathematics, University of Luxembourg \and \url{https://math.uni.lu/parlier/}}{hugo.parlier@uni.lu}{}{}

\author{Monique Teillaud}{Universit\' e de Lorraine, CNRS, Inria, LORIA, F-54000 Nancy, France \and \url{https://members.loria.fr/monique.teillaud/}}{monique.teillaud@loria.fr}{}{}

\authorrunning{V. Despr\' e, B. Kolbe, H. Parlier, M. Teillaud}

\funding{This work was partially supported by grant 
    ANR-17-CE40-0033 of the French National Research Agency ANR (project SoS)
    and
    INTER/ANR/16/11554412/SoS of the Luxembourg National Research fund FNR
    \url{https://SoS.loria.fr/}.}



\EventEditors{Erin Chambers and Joachim Gudmundsson}
\EventNoEds{2}
\EventLongTitle{39th Symposium on Computational Geometry}
\EventShortTitle{SoCG 2023}
\EventAcronym{SoCG}
\EventYear{2023}
\EventDate{June 12-15, 2023}
\EventLocation{Dallas, USA}
\EventLogo{}
\SeriesVolume{YY}
\ArticleNo{89}

\hideLIPIcs

\begin{document}
\maketitle

\begin{abstract}
The goal of this paper is to exhibit and analyze an algorithm that takes a given closed orientable hyperbolic surface and outputs an explicit Dirichlet domain. The input is a fundamental polygon with side pairings. While grounded in topological considerations, the algorithm makes key use of the geometry of the surface. We introduce data structures that reflect this interplay between geometry and topology and show that the algorithm finishes in polynomial time, in terms of the initial perimeter and the genus of the surface.
\end{abstract}

\section{Introduction and motivation}\label{sec:intro}

Hyperbolic surfaces and their moduli spaces play an ubiquitous role in mathematics, namely, through  relationships with other areas including Riemannian geometry, number theory, geometric group theory and mathematical physics. Algorithms for surface groups, as combinatorial or topological objects, have a rich history dating back to Dehn. Recently, in part motivated by applications in other sciences~\cite{bv-cp-86,nature-2022}, there has been a push to understand hyperbolic structures on surfaces from the point of view of computational geometry. 

Dealing with hyperbolic surfaces necessarily involves describing them --- or even visualizing them --- meaningfully. A fundamental domain (in the hyperbolic plane) with a side pairing is one way to determine a hyperbolic metric on the surface. Lengths of curves in a pants decomposition and their associated pasting parameters (so-called Fenchel-Nielsen coordinates) are another. No matter which construction or parameter set used, it is always interesting to know to which extent two different constructions output the ``same'' surface, where ``same'' can take different meanings. However, these representations, either by a fundamental domain or a set of Fenchel-Nielsen coordinates, are not unique, and determining a canonical representation is challenging for either option. In this paper, we tackle this question for fundamental domains, by computing a so-called Dirichlet domain. 

Roughly speaking, a Dirichlet domain of a hyperbolic surface is a fundamental polygon in the hyperbolic plane, with a special point where distances to that point in the polygon correspond to distances on the surface. Another way of thinking of them is that it is a Voronoi cell associated to a lift of a single point of the surface to its universal cover $\H$. A more precise definition is given in the next section. Note that for hyperbolic surfaces any given surface has infinitely many Dirichlet domains up to isometry. This is in strong contrast to, for example, flat tori. Nonetheless, when describing a surface via fundamental domains, the prize for the most relevant geometric domain undoubtedly goes to Dirichlet domains because they visualize the distance function for a given point. As far as we know, there is only one algorithm in the literature that computes a Dirichlet domain for a hyperbolic surface and a given point~\cite{Voight2009}. Unfortunately, the run-time of the algorithm is not studied and an analysis seems complicated.

The contribution of this paper is an algorithm that computes a Dirichlet domain efficiently, and its analysis. The point defining the domain is not given as input, but it is part of the output. The Dirichlet domain of a given input point can then be computed with a complexity that only depends of the genus of the surface~\cite{dkt-rihpd-21}.
Our main result is the following:
\begin{theorem}\label{th:main}
Let $\hypS$ be a closed orientable hyperbolic surface of genus $g$ given by a fundamental polygon of perimeter $L$ and side pairings. A Dirichlet domain for $\hypS$ can be computed in $O((g^2L)^{6g-4})$ time.
\end{theorem}

A key ingredient is the use of \emph{Delaunay triangulations} on hyperbolic surfaces, an area of research that has recently gained traction, both from an experimental and a theoretical perspective~\cite{Bogdanov2014,deblois2015,bogdanovsocg,iordanov,deblois2018,ebbens_et_al}.
Recently, it has been shown that the well-known flip algorithm that computes the Delaunay triangulation of a set of points in the Euclidean plane $\mathbb{E}^2$ also works on a hyperbolic surface; the complexity result announced in Theorem~\ref{th:main} crucially depends on the only known upper bound on the complexity of this Delaunay flip algorithm~\cite{despreflips}. The algorithm subsumes the real RAM model. Studying the algebraic numbers involved in the computations goes beyond the scope of this paper. 

The paper is structured as follows: In Section~\ref{sec:algo}, we give an overview of the algorithm and we present the data structure. Sections~\ref{sec:vertreduction} and~\ref{sec:magic} explain in detail the main two steps of the algorithm, which output a geometric triangulation of the surface having only one vertex. Finally, Section~\ref{sec:dirichlet} builds on the literature and concludes the proof of Theorem~\ref{th:main} with the last step of the algorithm.

\section{Preliminaries}\label{sec:prelim}

We begin by recalling a collection of facts and setting notations, and we refer to~\cite{beardon,busercompactsurfs,primermcgs} for details. The surfaces studied in this paper are assumed to be closed, orientable, and of genus $g\geq 2$. We begin with a topological surface and endow it with a hyperbolic metric to obtain a hyperbolic surface, generally denoted by $\hypS$. A hyperbolic surface is locally isometric to its universal covering space, the hyperbolic plane $\H$. Such surfaces can always be obtained by considering the quotient of $\H$ under the action of $\gp$, a discrete subgroup of isometries of $\H$ isomorphic to the fundamental group $\pi_1(\hypS)$. 

Let  $\hypS:=\H/\gp$ be a hyperbolic surface of genus $g$ and fundamental group $\gp$. The projection map is denoted as $\proj: \H \rightarrow \hypS=\H/\gp$. We denote by $\h{x}\in\proj^{-1}(x)$ one of the lifts, to $\H$, of an object $x$ on $\hypS$. More generally, objects in $\H$ are denoted with $\h{~}$.

A \emph{fundamental domain} $\fd$ for the action of $\gp$ is defined as a closed domain, i.e., $\overline{\inter(\fd)}=\fd$, such that $\gp\fd=\H$ and the interiors of different copies of $\fd$ under $\gp$ are disjoint.

For a point $\h{x}\in \H$, the \emph{Dirichlet domain} $\Dir{x}$ is defined as the Voronoi cell containing $\h{x}$, of the Voronoi diagram associated to the point set $\gp\h{x}$. In other words,
\[ \Dir{x}=\{\: \h{y}\in \H\: | \: d_{\H}(\h{x},\h{y})\le d_{\H}(\h{x},\gp \h{y})\: \}=\{\: \h{y}\in \H \: | \: d_{\H}(\h{x},\h{y})\le d_{\H}(\gp \h{x},\h{y})\:\},\]
where the equality is true since $\gp$ acts as isometries w.r.t.\ $d_{\H}$. 
The Dirichlet domain is a compact convex fundamental domain for $\gp$ with finitely many geodesic sides~\cite[\S9.4]{beardon} and is generally considered a canonical choice of fundamental domain. A property of Dirichlet domains, of interest for the conception of algorithms, is that, by the triangle inequality, 
\[
\dia(\Dir{x})\le 2\dia(\hypS)\le 2\dia(\Dir{x}),
\]
where $\dia(\cdot)$ denotes the diameter. 

\subsection{Curves, paths, and loops}\label{sec:curves}
Recall that a closed curve is the image of $\SSS^1$ under a continuous map; a curve is non-trivial (or essential) if it is not freely homotopic to a point. Similarly, a path is a continuous image of the interval $[0,1]$, and the images of $0$ and $1$ are referred to as its endpoints. A loop is a path whose endpoints are equal; this endpoint is referred to as its basepoint.

For a closed curve or loop $c$, we will denote by $[c]$ its free homotopy class, and, if $c$ is based in a point $p$, by $[c]_p$ its homotopy class of loops based in $p$. For a path $c$ between points $p$ and $q$, we denote by $[c]_{p,q}$ the homotopy class of the path with fixed endpoints. We will readily make use of the fact that if $c$ is closed non-trivial curve on a hyperbolic surface, then in $[c]$ there is a unique closed geodesic. Similarly, if $c$ is a loop based in $p$, in $[c]_p$ there is a unique closed geodesic loop, and if $c$ is a path between $p$ and $q$, in $[c]_{p,q}$ there is a unique geodesic path. If $c$ is a simple closed curve then the closed geodesic in $[c]$ is also simple, but this is no longer necessarily the case for loops or paths with basepoints.

The intersection number $i(c,c')$ between homotopy classes of curves $c$ and $c'$ is defined as the minimal intersection among its representatives. Note that closed geodesics on a hyperbolic surface always intersect minimally. The situation for paths is slightly different. The unique geodesic representatives of paths (with fixed end points) might not intersect minimally. This subtlety plays a key technical role in our story.

\subsection{Fundamental polygon}\label{sec:fpoly}
Let $\hypS$ be a (closed) hyperbolic surface of genus $g$ and fundamental group $\gp$. A polygon $P\subset\H$ (i.e., a circular sequence of geodesic edges) bounding a fundamental domain for $\gp$ (as defined in the introduction) is called a fundamental polygon. Poincaré's theorem implies that $\gp$ is generated by the side pairings on $P$~\cite[$\S9.8$]{beardon}.
The edges and vertices of $P$ project to a graph $G_P$ on $\hypS$; the region enclosed by $P$ projects to the unique face of $G_P$.

The numbers $n_G$ of vertices and $m_G$ of edges of $G_P$ satisfy Euler's relation $n_G-m_G+1=2-2g$, as there is only one face. It follows that if $G_P$ only has one vertex, then that vertex is incident to the $m_G=2g$ edges, which are actually all loops. The number of vertices is maximal when they all have degree~3 (then there are no loops); in this case $3n_G=2m_G$, so, $m_G=6g-3$ and $n_G=4g-2$. More generally, the number $2m_G$ of edges and vertices of $P$ lies between the two extreme cases: $4g\leq 2m_G \leq 12g-6$. When $2m_G < 12g-6$, some vertices of $P$ project to the same vertex of $G_P$, i.e., they belong to the same orbit under~$\gp$. $G_P$ has a loop for each edge whose vertices are in the same orbit; then the projected point on $\hypS$ is incident to that loop twice.

\section{Algorithm overview}\label{sec:algo}

Let $\hypS$ be a (closed) hyperbolic surface of genus $g$ and fundamental group $\gp$. 

We propose the algorithm sketched below to compute a Dirichlet fundamental domain of $\hypS$. The output of Step~\ref{step:vertreduction} will be denoted with primes; it will be used as input for Step~\ref{step:magic}, whose output will be denoted with double primes. 

\begin{enumerate}
\item\label{step:vertreduction} Construct a system $\beta'_0,\ldots,\beta'_{2g-1}$ of simple topological loops based at the same point~$b'$ that cuts $\hypS$ into a disk (Section~\ref{sec:vertreduction}).
\item\label{step:magic} Find a point $b''$ so that the system of geodesic loops based at $b''$, conjugate to the ones computed in Step~(\ref{step:vertreduction}), is embedded (Section~\ref{sec:magic}).
\item\label{step:dirichlet} Construct the Dirichlet domain of a lift $\h{b''}$ of $b''$ (Section~\ref{sec:dirichlet}). 
\end{enumerate}

Obviously, the complexity of the algorithm heavily depends on the data structure used to store the objects involved in the constructions.
As the algorithm actually operates in the universal covering space $\H$ of $\hypS$, it is natural to present the data structure in $\H$. We assume that, as input, we are given a fundamental polygon $\poly\subset\H$ for $\gp$, together with side pairings, as in Section~\ref{sec:fpoly}. The data structure described below is actually equivalent to a combinatorial map~\cite[Section~3.3]{mt-gs-01} on $\hypS$, enriched with geometric information. In particular, for each vertex $x$ of $\graph$ (the projection of $\poly$ onto $\hypS$, as in Section~\ref{sec:fpoly}), the sequence of edges around $x$ is ordered (edges that correspond to a loop appear twice). 

\paragraph*{Description of the input.} Let a representative $\h{e_i}, i=0,\ldots,m-1$ be chosen for each couple of paired edges of $\poly$ and denote as $\g_0,\ldots,\g_{m-1}$ the corresponding side pairings in $\poly$: the other edge of the couple is $\inv{\g_i}\h{e_i}$, where $\inv{\g_i}$ is the inverse of $\g_i$. We denote the set of the $2m$ edges of $\poly$ as $E_\poly$ and the set of its $2m$ vertices as $V_\poly$. 
We choose a representative $\h{v_j}, j=0,\ldots,n-1$ for each orbit of vertices of $\poly$; $n$ is the number of vertices of $\graph$.   


Each element of $\gp$ can be represented as a word on the alphabet $\alphabet=\{\id,\g_0,\ldots,\g_{m-1},\\ \inv{\g_0},\ldots,\inv{\g_{m-1}}\}$, where $\id$ denotes the identity in $\gp$. Here, letters of $\alphabet$ and the corresponding generators in $\gp$ are denoted by the same symbol; this should not cause any confusion. 

The data structure is roughly a doubly linked list of edges of $\poly$, which stores the combinatorial information. Additional information is necessary to store the geometry (i.e., the positions of the vertices of $\poly$ in $\H$) and the side pairings. The data stored for each edge and vertex is constant, so the size of the data structure is $O(g)$ (we do not try to shave constants in the $O( )$).

Concretely, for each edge $\h{x}\in E_\poly$, the data structure stores: 
\begin{itemize}
\item two pointers $\prev{\h{x}}$ and $\nxt{\h{x}}$ that give access to the previous and next edges in $\poly$, respectively (in counterclockwise order); 
\item two pointers $\source{\h{x}}$ and $\target{\h{x}}$ that give access to the source and target of $\h{x}$ in $\poly$, respectively (in counterclockwise order); when $\proj\h{x}$ is a loop in $\graph$, $\source{\h{x}}$ and $\target{\h{x}}$ lie in the same orbit under $\gp$;
\item a pointer to the paired edge $\pair{\h{x}}$ in $\poly$; 
\item a letter $\alp{\h{x}}\in\alphabet$ that encodes the relation between $\h{x}$ and $\pair{\h{x}}$: \\
$ \alp{\h{x}} = 
\begin{dcases*}
\id & if $\h{x}=\h{e_i}$ \\
\g_i & if $\h{x}=\inv{\g_i}\h{e_i}$ 
\end{dcases*}
$ for some $i\in\{0,\ldots,m-1\}$.

By definition, 
$ \pair{\h{x}} =
\begin{dcases*}
\inv{\g_i}\h{x} & when $\alp{\h{x}}=\id$ ($\h{x}=\h{e_i}$)\\
\g_i\h{x} & when $\alp{\h{x}}=\g_i$
\end{dcases*}
$.
\end{itemize}

For each vertex $\h{y}\in V_\poly$, the data structure stores:
\begin{itemize}
\item  $\point{\h{y}}$, which is the representative point of its orbit: $\point{\h{y}}=\h{v_j}$ for some $j\in\{0,\ldots,n-1\}$; 
\item a word $\g_{\h{y}}$ on $\alphabet$ (equivalently, $\g_{\h{y}}\in\gp$), which specifies the precise position $\g_{\h{y}}\,\point{\h{y}}$ of $\h{y}$ in $\H$. 
\end{itemize}

The graph $\graph$ lifts in the universal covering space $\H$ to the (infinite) graph $\proj^{-1}{\graph}=\gp\poly$. 
In particular, the sequence of edges of $\gp\poly$ incident to a given vertex $\h{v}\in\proj^{-1}{v}$ is a sequence of lifts of the edges incident to $v$ in $\graph$. Each of these lifts is the image by an element of $\gp$ of an edge of $\poly$ (Figure~\ref{fig:vertex}). The following result is straightforward from the data structure. We still prove it for completeness.

\begin{lemma}\label{lem:incident-edges}
Let $e$ be an edge of $\graph$ and $v$ a vertex of $e$. The sequence of edges of $\graph$ incident to $v$ can be found in time $O(g)$.
\end{lemma}
\begin{proof}
Without loss of generality (this can always be achieved by renaming edges and vertices of $\poly$) $\h{v}=\source{\h{e}}$ (as in Figure~\ref{fig:vertex} for $e=e_0$), or $\h{v}=\target{\h{e}}$. 
Consider the first case. After $e$, the next edge incident to $v$ in counterclockwise order in $\graph$ is given in $\poly$ by $\h{x}=\prev{\h{e}}$, whose vertex $\target{\h{x}}$ is $\h{v}$. The next edge incident to $v$ in $\graph$ is given by $\prev{\pair{\h{x}}}$, whose target vertex lies in the same orbit as $\h{v}$ under $\gp$. And so on: a sequence of accesses to $\pair{\cdot}$ and $\prev{\cdot}$ allows us to find the edges of $\gp\poly$ incident to $\h{v}$ in counterclockwise order. The process perfoms a constant number of accesses for each edge incident to $\h{v}$, and the number of such edges is linear in $g$ as recalled above. 
The case when $\h{v}=\target{\h{e}}$ is similar: $\nxt{\cdot}$ is used instead of $\prev{\cdot}$.
\end{proof}

In addition, the precise positions in $\H$ of all vertices of $\poly$ in the orbit $\proj^{-1}{v}$ can be computed along the process using the information $\point{\cdot}$ and $\alp{\cdot}$ stored in the data structure, without changing the complexity. 

Relations in the finitely presented group~\cite[Chapter~5.5]{cm-grdg-57} $\gp$ can be deduced by comparing the two sequences of edges ---clockwise and counterclockwise--- around each vertex. 

\begin{figure}[htbp]
	\centering
	\includegraphics[page=1]{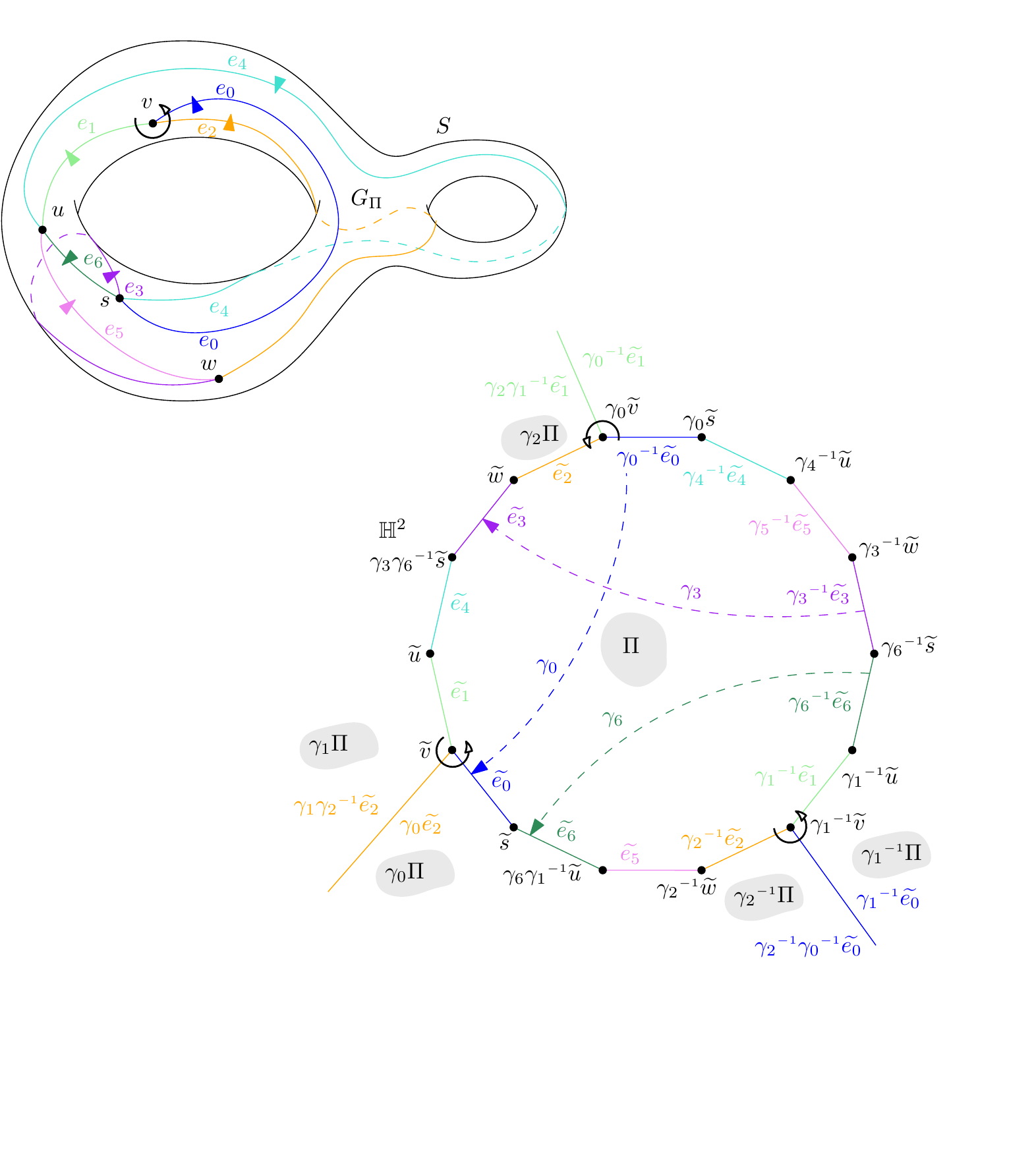}
\caption{(Top) The graph $\graph$. The arrow around vertex $v$ shows its incident edges. (Bottom) The fundamental polygon $\poly$. Vertices $\h{s},\h{u},\h{v}$, and $\h{w}$ of $\poly$ are chosen as representatives of the orbits of $s,u,v$, and $w$, respectively. The arrows show the combinatorics of the tiling $\gp\poly$ at the three vertices of $\poly$ in the orbit $\proj^{-1}{v}$: $\h{v}, \g_0\h{v}=\g_2\inv{\g_1}\h{v}$, and $\inv{\g_1}\h{v}=\inv{\g_2}\g_0$. }
	\label{fig:vertex}
\end{figure}

\section{Constructing the initial system of simple loops}\label{sec:vertreduction} 

The combinatorial part of Step~\ref{step:vertreduction} of the algorithm is quite common in the topology literature: it consists in computing a spanning tree $\tree$ of $\graph$, then the edges of $\tree$ are contracted, so that each vertex of $\tree$ is merged into the root, and each edge of $\graph$ that is not an edge of $\tree$ is transformed into a loop based at the root. This is illustrated in genus~2 by Figure~\ref{fig:tree}(Top). In this example, $\tree$ has three edges $e_1,e_5,$ and $e_6$. If $v$ is chosen as the root, edge $e_0$ transforms into a loop based at $v$ when $e_1$ and $e_6$ are contracted. 

However, topology is not enough in this work. We actually compute the geometry of each loop that is obtained from the contraction of $\tree$ by precisely computing a lift. 

The main result of this section is as follows:
\begin{proposition}\label{prop:initial-loops}
Let $\hypS$ be a closed orientable surface of genus $g$ and $\poly$ a fundamental polygon of $\hypS$ with $2m$ edges and side pairings as described in Section~\ref{sec:algo}. 
A system of loops based at a common point on $\hypS$, given by a circular list of geodesic segments in $\H$ and side pairings, can be constructed in time $O(g^3)$. The total length of this system of loops is $O(g L)$, where $L$ denotes the perimeter of $\poly$. 
\end{proposition}

The construction algorithm proceeds in three phases:
\renewcommand{\theenumi}{\it(\roman{enumi})}
\begin{enumerate}
\item \label{step:mst} Compute a spanning tree $\tree$ of $\graph$. A root $b$ is chosen for $\tree$, together with an edge $e_0$ incident to $b$ in $\graph\setminus\tree$ and lifts $\h{b}$ and $\h{e_0}$ in $\poly$. Up to a renaming of representatives in orbits, we can assume that $\h{b}=\source{\h{e_0}}$.
\item \label{step:newfp} Construct a new fundamental domain $\poly'$, as a polygon whose edges are paths consisting of $O(g)$ geodesic segments in $\gp\poly$: in each such path, one segment is a lift of an edge of $\graph\setminus\tree$, and the other segments are lifts of edges of $\tree$; the endpoints of each path lie in the orbit of $\h{v}$. The side pairings in $\poly'$ are also computed. 
\item \label{step:loops} Replace each path computed in the previous stage by the geodesic segment between its vertices and keep the side pairings. 
\end{enumerate}

Note that $\poly'$ is a fundamental domain, but not a fundamental polygon in the sense of Section~\ref{sec:fpoly}: its edges are paths consisting of several geodesic segments; the geodesic segments between its vertices (i.e., endpoints of these paths) will intersect in general, so they do not bound a fundamental domain. We will call such a polygon a \emph{topological polygon}. Section~\ref{sec:magic} will present the construction of a fundamental polygon from this topological polygon (Step~\ref{step:magic} of the algorithm).

The rest of this section is devoted to proving Proposition~\ref{prop:initial-loops}, by detailing the construction.

\begin{proof}
As in Section~\ref{sec:algo}, $n$ denotes the number of vertices of $\graph$. Phase~\ref{step:mst} is performed by a standard constuction of a minimum spanning tree $\tree$ in $O(m\log n)$ or $O(m+n\log n)$, i.e., $O(g\log g)$. The tree has $n-1$ edges. 

Phase~\ref{step:newfp} consists of walking along the edges of $\gp\poly$. The walk constructs the new fundamental domain $\poly'$ in counterclockwise order and stores it in a data structure that is very similar to the data structure defined in Section~\ref{sec:algo} for $\poly$. However some of its elements have a different meaning, which will be detailed in the sequel; in particular, the elements $\pair{\cdot}$ are actually not yet side pairings in this phase, but temporary elements of $\gp$. 

As a preprocessing step, for each edge $x$ of $\graph\setminus\tree$, we find the path $\pat{x}$ on $\hypS$ whose homotopy class contains the loop that will eventually replace $x$: it is given by the (unique and simple) path in $\tree$ from the root to a first vertex of $x$, followed by $x$, and finally by the path in $\tree$ from the second vertex of $x$ to the root. In the example of Figure~\ref{fig:tree}(Top), $e_0$ is replaced by a loop based at $b=v$ that is homotopic to the sequence $\pat{e_0}=e_0\cdot e_6\cdot e_1$, where~$\cdot$ denotes concatenation of paths. The path for edge $e_4$ is $\pat{e_4}=e_1\cdot e_4 \cdot e_6 \cdot e_1^{-1}$; here, edge $e_1$ is traversed in both directions.

The walk starts at $\h{b}$ and first considers edge $\h{e_0}$ chosen in stage~\ref{step:mst}. For each considered edge $\h{x}$ not in $\proj^{-1}{\tree}$, by the pre-processing we have just mentioned, we look for lifts of edges of $\pat{x}$ in order in $\gp\poly$. This is easily done by a sequence of operations $\nxt{\cdot},\prev{\cdot}$, and $\pair{\cdot}$ on edges of $\poly$, and turning around their vertices $\source{\cdot}$ and $\target{\cdot}$ as in Lemma~\ref{lem:incident-edges} until a lift of the next element of $\pat{x}$ is found. On the way, the elements $\alp{\cdot}$ of $\gp$ found in the data structure are collected so that the precise lift of each edge or vertex of $\poly'$ is known. 

Each time a lift of a path $\pat{x}$, i.e., an edge of $\poly'$, has been found, the algorithm proceeds to the next one. Note that edges (i.e., paths) appear on $\poly'$ in the same order as the order in which the corresponding edges appear on $\poly$: indeed, contracting the edges of $\tree$ does not change the order in which edges on $\hypS$ are traversed to describe the boundary of the face of~$\graph$. 

\begin{figure}[htbp]
	\centering
	\includegraphics[page=2]{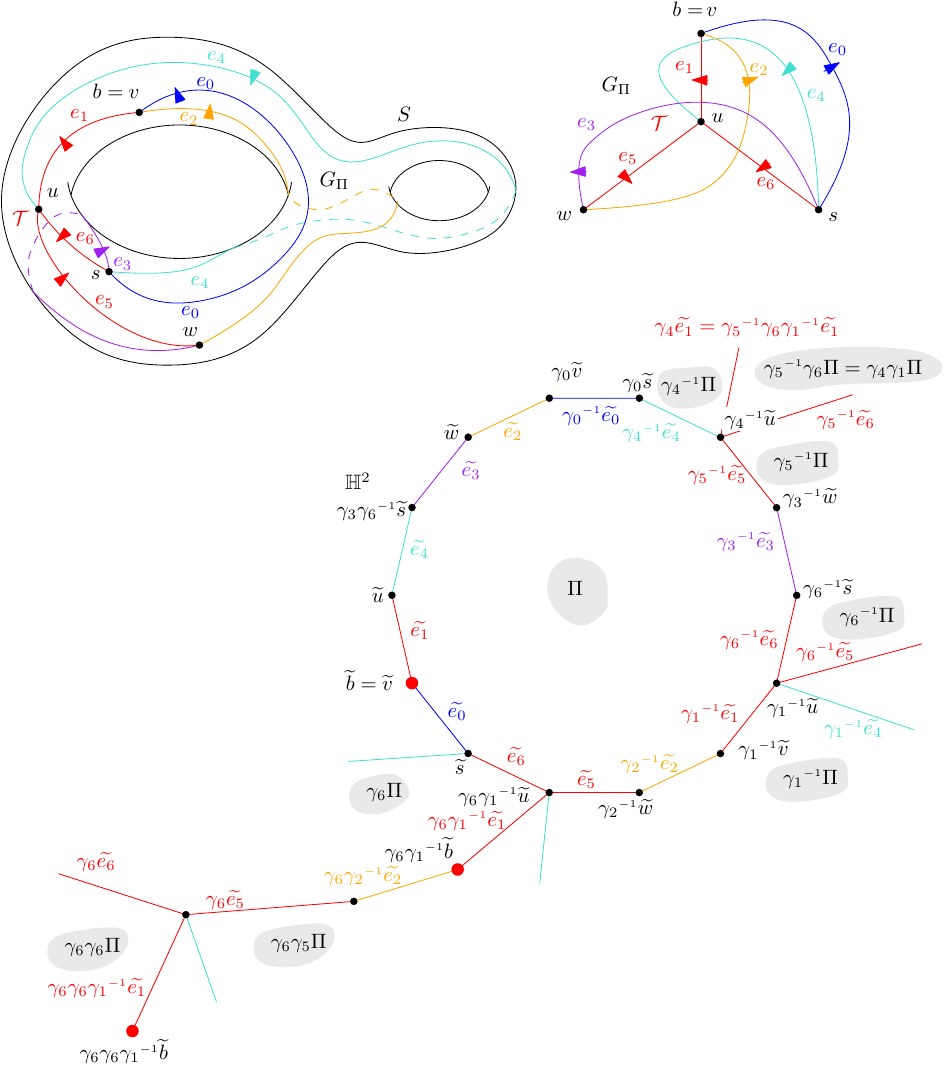}
	\caption{(Top) The spanning tree $\tree$ of $\graph$ has edges $e_1,e_5,e_6$ and is rooted at $v=b$. Each edge of $\graph\setminus\tree$ is replaced by a loop based at $b$ when contracting $\tree$. 
                     (Bottom) The path in $\H$ corresponding to the loop replacing $e_0$ starts at $\h{b}$ and ends at $\g_6\inv{\g_1}\h{b}$.}
	\label{fig:tree}
\end{figure}

This is illustrated in Figure~\ref{fig:tree}(Bottom). The walk starts from $\h{b}$ and follows $\h{e_0}$. The next edge in $\poly$ is $\nxt{\h{e_0}}=\h{e_6}$, which projects onto the next edge in $\pat{e_0}$. Then we must look for a lift of $e_1$ incident to $\target{\h{e_6}}$. This is done by going to $\pair{\h{e_6}}=\inv{\g_6}\h{e_6}$ and turning around its source vertex. The first edge in counterclockwise order is $\inv{\g_1}\h{e_1}$; the source vertex of its image $\g_6\inv{\g_1}\h{e_1}$ is the target vertex of $\h{e_6}$ and the walk traverses it. Its target vertex is $\g_6\inv{\g_1}\h{b}$, which is in the orbit of $\h{b}$. We have now found the lift of $\pat{e_0}$ in $\gp\poly$ that forms the first edge of $\poly'$: it is the sequence $\h{e_0} \cdot \h{e_6} \cdot \g_6\inv{\g_1}\h{e_1}$. From vertex $\g_6\inv{\g_1}\h{b}$ we will now construct the edge of $\poly'$ corresponding to $\h{e_2}$, as $\h{e_2}$ is the edge following $\inv{\g_0}\h{e_0}$ in $\poly\setminus\proj^{-1}{\tree}$. We know that $\pat{e_2}=e_2 \cdot e_5 \cdot e_1$. Here, turning around the source vertex of $\inv{\g_1}\h{e_1}$ gives $\inv{\g_2}\h{e_2}$, and the walk continues with $\g_6\inv{\g_2}\h{e_2}$, then turning around the source of $\inv{\g_2}\h{e_2}$ we find $\h{e_5}$, and the walk follows $\g_6\h{e_5}$. So far we have only followed edges of $\g_6\Pi$, as the edge that we were looking for when turning around vertices was always the first one. However, this is not the case after $\g_6\h{e_5}$. The target of the representative $\inv{\g_5}\h{e_5}$ is $\inv{\g_4}\h{u}$, around which we must turn until we find a lift of $e_1$; the next edge of $\gp\Pi$ that we follow is thus $\g_6\g_5\inv{\g_5}\g_6\inv{\g_1}\h{e_1}=\g_6\g_6\inv{\g_1}\h{e_1}$, which finishes the edge of $\poly'$ corresponding to $\h{e_2}$. Next, we would continue with $\h{e_3}$ in the same vein. And so on. 

Note that, as we are constructing the fundamental domain $\poly'$, following the order of the edges of $\poly$, each edge $\h{e'_i}$ of $\poly'$ defines a topological loop $\beta_i'$ based at $b$ on $\hypS$, which represents the homotopy class $[e'_i]$. Such an edge $\h{e'_i}$ is formed by a sequence of edges of $\gp\poly$ that corresponds to the path $\pat{e_i}$, for $e_i\in\graph\setminus\tree$, and will naturally be paired with another sequence for the same $\pat{e_i}$ (traversed in the opposite direction around $\poly'$). The words associated with the edges in the two sequences differ by an element of $\gp$, which gives the side pairing $\gamma_i'\in\gp$ for $\poly'$. 

Phase~\ref{step:loops} is easy. It consists in replacing each edge of $\poly'$ by the geodesic segment between its two vertices, and keeping the associated side pairings. As mentioned above, the corresponding geodesic loops may intersect on $\hypS$, though the topological loops that we choose to represent their homotopy classes only intersect at their common basepoint. 

As the edges of $\poly'$ project by construction to loops, all based at the same point, there are $2g$ such loops on $\hypS$ and $\poly'$ has $4g$ edges, each consisting of $O(g)$ edges (and vertices) of $\gp\poly$. By Lemma~\ref{lem:incident-edges}, $O(g)$ operations are performed at each vertex. This shows the complexity annonced in Proposition~\ref{prop:initial-loops}. The bound on the sum of the lengths of the geodesic loops also follows directly.
\end{proof}

Note that during the traversal detailed in the proof, we have computed for each vertex $\h{x}$ of $\poly'$ the element $\g\in\gp$ such that $\h{x}=\g \h{b}$. We store these elements of $\gp$ in a table $\tab$, which will be used in the sequel, in addition to the data main data structure. 

We denote the output of this step~\ref{step:magic} as follows: we re-index the sides of the topological polygon $\poly'$ (which has $4g\leq 2m$ edges) so that the side pairings are denoted as $\g'_0,\ldots,\g'_{2g-1}$ and the corresponding $2g$ topological loops on $\hypS$ are $\beta'_0,\ldots,\beta'_{2g-1}$; these loops on $\hypS$ do not intersect except at their common basepoint $b$, which is now renamed to $b'$ for global consistency of notation, as announced at the beginning of Section~\ref{sec:algo}.

\section{Finding an embedded system of loops}\label{sec:magic}
We want to find a collection of geodesic loops on a hyperbolic surface $\hypS$, all based in a single point and disjoint otherwise, and such that the complementary region of the loops is a convex hyperbolic polygon. What we show is that in fact we can retain the choice of topological loops $\beta'_0,\ldots,\beta'_{2g-1}$ made in Section~\ref{sec:vertreduction} by moving the basepoint appropriately to ensure that their geodesic realizations satisfy the desired properties.

Consider the set of topological loops $\beta'_0,\hdots, \beta'_{2g-1}$ all based at point $b'$ constructed in the previous section. We choose a pair that intersects minimally exactly once which, up to reordering, we can suppose are $\beta'_0$ and $\beta'_1$. For future reference we set $L_0:= \max\{ \ell(\beta'_0), \ell(\beta'_1)\}$, where $\ell$ denotes the length. 
\begin{remark}\label{rem:inter}
We can fix any loop to be $\beta'_0$ and find a loop $\beta'_1$ intersecting it exactly once. Indeed, the set $\beta'_0,\hdots, \beta'_{2g-1}$ contains curves that pairwise intersect at most once, and are all non-separating and thus homologically non-trivial. As it generates homotopy, it also generates homology and in particular every curve must be intersected by at least one other curve. As they can intersect at most once, they intersect exactly once. 
\end{remark} 

We begin by taking the unique geodesic loops, based in $b'$, in the free homotopy classes of $\beta'_0$ and $\beta'_1$, and we replace the curves with these geodesic representatives (we keep the same notation for convenience). Now we further consider the unique simple closed geodesic representatives in the free homotopy class of $\beta'_0$ and $\beta'_1$, which we denote $\beta''_0$ and $\beta''_1$, respectively. By hypothesis, they intersect in a single point $b''$, which will be our new basepoint. 

We now define a path between $b''$ and $b'$ as follows. We consider a single lift $\h{\beta'_0}$ of $\beta_0'$. Its endpoints both correspond to distinct lifts of $b'$ which are related by a unique translation $g_0$ in $\gp$. The copies of $\h{\beta'_0}$ by iterates of $g_0$ form a broken geodesic line with the same end points at infinity as the geodesic axis of $g_0$. This singular geodesic, which we denote $\hat{\beta_0'}$, separates $\H$ into two half-spaces, only one of which is convex. We now choose an endpoint of $\h{\beta'_0}$ and consider a lift of $\beta_1'$ that lies in the convex half-space. This lift we denote by $\h{\beta'_1}$ and, as before, we consider the corresponding translation $g_1$ in $\gp$ and its geodesic axis and its corresponding singular geodesic $\hat{\beta_1'}$. Now, we obtain $\h{b''}$ as the intersection of the axes of $g_0$ and $g_1$. We consider the unique geodesic path $\h{c}$ between $\h{b'}$ and $\h{b''}$ and its projection $c$ on $\hypS$. We first observe that we can control the length of this path $c$:

\begin{lemma}\label{lem:betalength}
	$\ell(c) < 2 L_0$.
\end{lemma}

\begin{proof}
We observe that the axis of $g_0$ must lie in an $R$ neighborhood of $\hat{\beta_0'}$ where $R < \ell(\beta_1')$. In particular, the axis of $g_1$ intersects $\h{\beta'_0}$. Similarly, the axis of $g_0$ intersects $\h{\beta'_1}$. 
Now the proof essentially follows from drawing a picture of the above situation in $\H$ (see Figure \ref{fig:photo}). 

\begin{figure}[htbp]
	\centering
	\includegraphics[page=1]{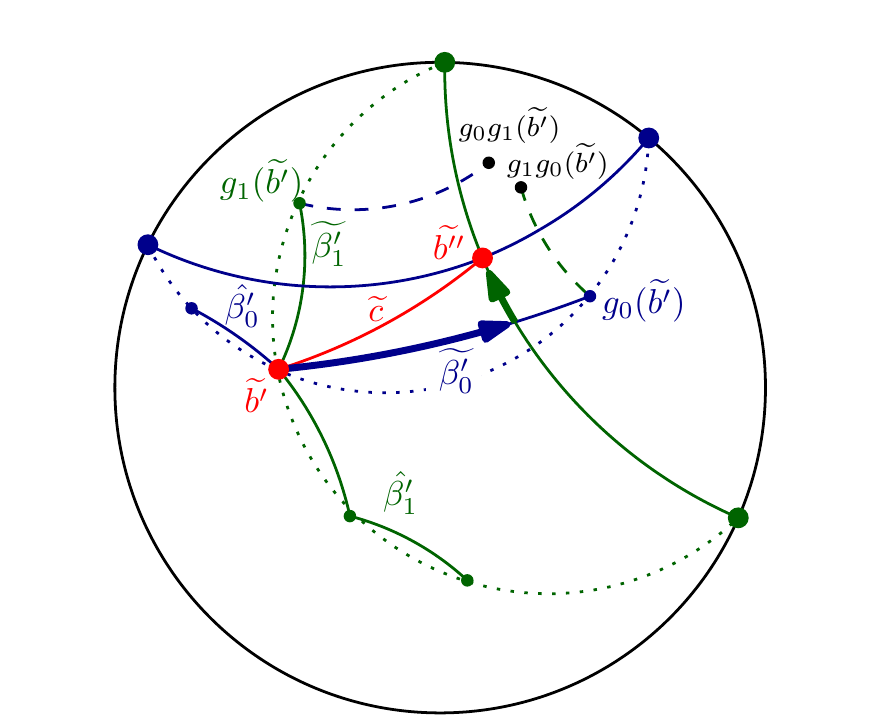}
	\caption{$c$ is homotopic on $\hypS$ to the projection of the concatenation of the bold arc of $\h{\beta'_0}$ from $\h{b'}$ and the bold segment of the axis of $g_1$.}
	\label{fig:photo}
\end{figure}

By following an arc of $\h{\beta'_0}$ from $\h{b'}$ and then a segment of length at most $\ell(\h{\beta''_1})$ on the axis $g_1$, we obtain a path between $\h{b'}$ to $\h{b''}$. 
As such, we have
\[
\hspace{13em}\ell(\h{c}) < \ell(\h{\beta'_0}) +  \ell(\h{\beta''_1})
\]
and so by passing to the surface
\[
\hspace{15em}\ell(c) < 2 L_0.
\]
\end{proof}

Observe that for $i=0,1$, $\beta''_i$, based in $b''$, is freely homotopic to $c^{-1} \cdot \beta'_i \cdot c$ and that there is a homeomorphism of $\hypS$, isotopic to the identity, that takes $b'$ to $b''$ and that sends (the homotopy class of) $\beta'_i$ to $\beta''_i$. This homeomorphism is often referred to as the point pushing map (see for instance \cite[Section~4.2]{primermcgs} for details).
	
We can apply this same homeomorphism to the remaining curves. For $i=0,\hdots, 2g-1$ we set the homotopy class of loop $\beta''_i$ to be:
\begin{equation}\label{eq:conjug}
[\beta''_i]_{b''}= [c^{-1} \cdot \beta'_i \cdot c]_{b''}.
\end{equation}
As we have just moved the basepoint by a homeomorphism, the homotopy classes $[\beta''_i]_{b''}$ all have simple representatives and can be realized disjointly outside of $b''$. The following lemma implies that their unique geodesic representatives enjoy this same property. It is well known to specialists, but we include a proof sketch for completeness. 

\begin{lemma}\label{lem:convex}
	Let $\Sigma$ be a hyperbolic surface with piecewise-geodesic boundary such that the interior angles on the singular points $s_0,\hdots,s_{k-1}$ of the boundary are cone points of angle $\leq \pi$. If $[\alpha]_{p_i,q_i}, [\alpha']_{p_j,q_j}$ are simple homotopy classes of paths (with endpoints $p_i,p_j,q_i,q_j$ in the set $s_0,\hdots,s_{k-1}$), and disjoint except for possibly in their endpoints, then the unique geodesic representatives are also simple and disjoint. 
\end{lemma}

\noindent\emph{Sketch of proof.}
	We consider $\h{\Sigma}$, the universal cover of $\Sigma$, which we view as a (geodesically convex) subset of $\H$. We lift $\partial \Sigma$ to $\h{\Sigma}$ and representatives of $[\alpha]_{p_i,q_i}$ and $[\alpha']_{p_j,q_j}$, which are simple and disjoint, to the universal cover. Observe that being simple and disjoint is equivalent to {\it all} individual lifts in $\H$ being simple and pairwise disjoint. Now take two individual lifts of either $\alpha$ or $\alpha'$, and their unique geodesic representatives. We will see that they are also disjoint. Note that in general, given two simple disjoint paths in the hyperbolic (or Euclidean) plane, the unique geodesics between their endpoints might intersect (as already mentioned in Section~\ref{sec:curves}). However:
	
	\noindent{\underline{Observation:}} {\it Let $C\subset \H$ be a convex with non-empty boundary, and  $p_0, q_0, p_1, q_1 \in \partial C$. Let $\alpha_1:[0,1]\to C$ and $\alpha_2:[0,1]\to C$ be simple paths, disjoint in their interior, with $\alpha_0(0)=p_0, \alpha_0(1)=q_0$ and $\alpha_1(0)=p_1, \alpha_1(1)=q_1$. Then the unique geodesic between $p_0$ and $q_0$ and the unique geodesic between $p_1$ and $q_1$ are disjoint in their interior as well.}
		
	A key point is that, thanks to the angle condition on the cone points, $\h{\Sigma}$ is a convex region of $\H$. (This is just a slightly more sophisticated observation than the elementary fact that a polygon with all interior angles less than $\pi$ is convex.) The observation now implies that the lifts of geodesics corresponding to $\alpha$ and $\alpha'$ are disjoint in their interior if and only if there are representatives of $[\alpha]_{p_i,q_i}$ and $[\alpha']_{p_j,q_j}$ that are, too, which, by hypothesis, is the case. 
\hfill $\qedsymbol$

\medskip

We can now apply Lemma \ref{lem:convex} to the geodesic representatives of $[\beta''_i]_{b''}$. For simplicity we denote by $\beta''_i$ the unique geodesic loop in the corresponding homotopy class.

\begin{theorem}\label{thm:betagood} Let $\beta''_0,\hdots,\beta''_{2g-1}$ be a set of topological loops based in $b''$ that cuts a surface $\hypS$ into a disk. Assume that $\beta''_0$ and $\beta''_1$ are closed geodesics. Then, the geodesic loops homotopic to $\beta''_0,\hdots,\beta''_{2g-1}$ are simple and pairwise disjoint in their interiors. Furthermore, by cutting $S$ along those geodesics and lifting to $\H$, one obtains a convex hyperbolic polygon with $4g$ edges.
\end{theorem}
    
\begin{proof}	
	As $\beta''_0$ and $\beta''_1$ are closed geodesics, they form $4$ angles in $b''$, and the opposite ones are equal. These angles thus satisfy $2\theta+2 \theta'= 2\pi$ so in particular both $\theta$ and $\theta'$ are strictly less than $\pi$. Thus by cutting along $\beta''_0$ and $\beta''_1$, we obtain a genus $g-1$ surface with a boundary consisting of $4$ geodesic segments, and with $4$ cone point singularities of angles~$<\pi$ (see Figure~\ref{fig:cutting}). 
	
	\begin{figure}[htbp]
		\centering
		\includegraphics[width=0.9\textwidth]{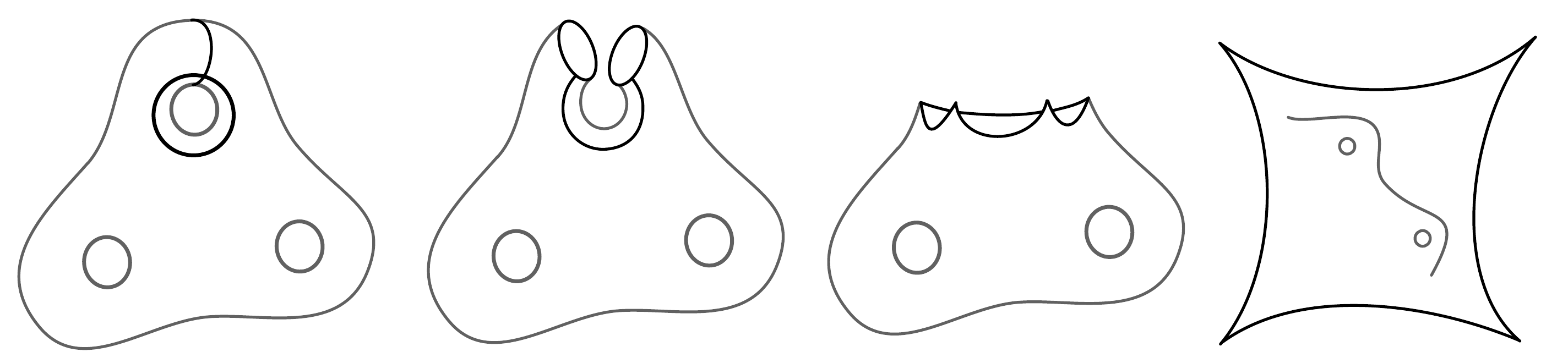}
		\caption{A visualization of the cutting along $\beta''_0$ and $\beta''_1$.}
		\label{fig:cutting}
	\end{figure}
	
	We now proceed inductively for $i\geq 3$ and consider the unique geodesic path $\beta''_i$, which by virtue of Lemma \ref{lem:convex}, has disjoint interior from the previous geodesic segments. Furthermore, as each segment further splits the angles, the angles are all less than $\pi$. 
	
	The end result is a polygon with all interior angles less than $\pi$ which, by elementary hyperbolic geometry, is convex.\end{proof}

\begin{proposition}\label{prop:convex-polygon}
	Let $\hypS$ be hyperbolic of genus $g$ and $\poly'$ a topological fundamental polygon of $\hypS$ with $4g$ edges and side pairings as described at the end of Section~\ref{sec:vertreduction}. 
	A convex fundamental polygon $\poly''$ with its side pairing and whose vertices project to a single vertex on $\hypS$, can be constructed in $O(g)$ time. The perimeter of $\poly''$ is $O(g L')$, where $L'$ denotes the total length of the sides of $\poly'$. 
      \end{proposition}
      
\begin{proof}
  We need to compute the output convex polygon $\poly''$ i.e., $4g$ lifts of $b''$ and $2g$ side pairings $\g''_0,\cdots,\g''_{2g-1}$. As homotopy classes of $\beta''_i$ and $\beta'_i$ are conjugates for $i=0,\ldots,2g-1$ (Equation~\ref{eq:conjug}), the side pairing $\g''_i$ is equal to $\g''_i$ for each $i$.

The key point here is the computation of a lift of $b''$. The first step consists in finding the loops $\beta'_0$ and $\beta'_1$ satisfying $i(\beta'_0,\beta'_1)=1$. As shown in Remark~\ref{rem:inter}, we can choose any loop for $\beta'_0$. We also fix $\h{b_0'}$ to be an endpoint of one of the two paired sides of $\poly'$ that are lifts of $\beta'_0$. We compute the ordered sequence of loops around $b'$ as in Lemma~\ref{lem:incident-edges}, in $O(g)$ operations; recall that each loop $\beta'_0,\hdots, \beta'_{2g-1}$ appears twice in the sequence (Section~\ref{sec:fpoly}). We take as $\beta'_1$ one of the loops that alternate with $\beta'_0$ in the sequence, and choose for $\h{\beta'_1}$ one of its two lifts that are incident to $\h{b_0'}$. 

The second step consists in finding the free geodesics in the homotopy classes of $\beta'_0$ and $\beta'_1$, respectively. In $\tab$, we find the word $g_0$ on $\{\g'_0,\cdots,\g'_{2g-1}\}$ representing the translation that sends $\h{b_0'}$ to the other endpoint of $\h{\beta'_0}$ (see Figure~\ref{fig:primetosecond}). The sequences $g_0^n(\h{b'_0})$ and $g_0^{-n}(\h{b'_0})$ converge in $\mathbb{C}$ to two points on the unit circle: these points are the two (infinite in $\H$) fixed points of the translation $g_0$, i.e., the two solutions of equation $g_0(z)=z$ in $\mathbb{C}$. The axis of $g_0$, i.e., the geodesic between these two points,  projects onto $\hypS$ to the free geodesic in $[\beta'_0]$. 

\begin{figure}[htbp]
	\centering
	\includegraphics[page=2]{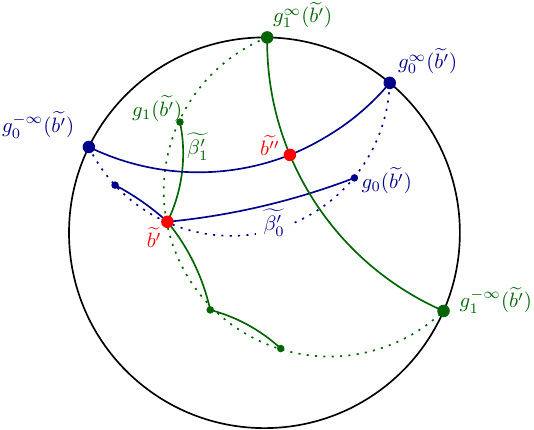}
	\caption{The computation of $b''$ from $b'$.}
	\label{fig:primetosecond}
\end{figure}

We repeat the same process with $\h{\beta'_1}$ and find the geodesic in $\H$ that projects to the free geodesic in $[\beta'_1]$. The point $\h{b_0''}$ comes as the intersection point of the two geodesics in $\H$. We now define $\h{\beta''_0}$ as the geodesic segment between $\h{b_0''}$ and $w_0(\h{b_0''})$, and $\h{\beta''_1}$ in the same way. This step is performed in constant time. 

We can now compute the $4g$ lifts of $b''$ that are the vertices of $\poly''$ by applying the elements of $\tab$ to $\h{b_0''}$. This last step has complexity $O(g)$. Additionally, we have 
\[
	\hspace{10em}\ell(\beta_i'')\le\ell(\beta_i')+2\ell(c)=5\cdot(\max_j\ell(\beta_j'))
\]
for each $i=1,\ldots,2g-1$ (by Lemma~\ref{lem:betalength}) and thus the perimeter of $\poly''$ is $O(g)$ times bigger than the perimeter $\poly'$.
\end{proof}

\section{Finding a Dirichlet domain from an embedded system of loops}\label{sec:dirichlet}
We first summarize what we have obtained so far. We started with a polygon $\poly$ of perimeter $L$ and we obtained a convex polygon $\poly''$ of total length $O(g^2L)$. Additionally, all vertices of $\poly''$ project on a single vertex $b''$ on $\hypS$. This construction has complexity $O(g^3)$ by Propositions~\ref{prop:initial-loops} and~\ref{prop:convex-polygon}. 
At this point, it is easy to compute a Dirichlet domain. Indeed, we can now triangulate $\poly''$ easily since it is convex and, thus, we obtain a geometric triangulation $T$, on to which the Delaunay flip algorithm can be applied~\cite{despreflips}.
The complexity of this algorithm depends on the diameter of $T$, for which the perimeter of $\poly''$ is an upper bound.

The output of the flip algorithm is a Delaunay triangulation $DT$ of $\hypS$ with the single vertex $b''$ computed in Section~\ref{sec:magic}. To obtain a Dirichlet domain from $DT$, we just have to compute the triangles of $\h{DT}$ incident to a lift $\h{b''}$ of $b''$ and their dual: we compute the circumcenter of each triangle to obtain the vertices of the Dirichlet domain and we put a geodesic between vertices that correspond to adjacent triangles around $\h{b''}$. This step is also clearly done in $O(g)$ operations. Putting all together we obtain the following theorem:

\begin{theorem}\label{theo:final}
	Let $\hypS$ be a closed orientable hyperbolic surface of genus $g$ given by a fundamental polygon of perimeter $L$ and side pairings. A Dirichlet domain of $\hypS$ can be computed in time $O(f(g^2L)+g^3)$ where $f(\Delta)$ is the complexity of the flip algorithm for a triangulation of diameter $\Delta$ with a single vertex.
      \end{theorem}

Using the best known bound for the flip algorithm so far, we obtain Theorem~\ref{th:main} stated in the introduction as a corollary. Note that the constant in the $O()$ depends on the metric on $\hypS$. However, there are experimental and theoretical insights suggesting that the actual complexity of the flip algorithm may be much better~\cite{loicflips2021}.


\bibliography{bibliography}

\end{document}